\documentclass[11pt]{article}

\usepackage[margin=1in]{geometry}
\usepackage{amsmath, amsthm, amssymb, amsfonts}
\usepackage{hyperref}
\usepackage{enumerate}
\usepackage{hyperref}

\newtheorem{theorem}{Theorem}
\newtheorem{lemma}[theorem]{Lemma}
\newtheorem{definition}{Definition}

\newtheorem{remark}[theorem]{Remark}
\newcommand{\eps}{\ensuremath{\varepsilon}}

\usepackage{graphicx} 
\usepackage{algpseudocode}
\usepackage{algorithm}

\usepackage[inline]{enumitem}

\usepackage{subcaption}

\usepackage{tikz}
\usetikzlibrary {arrows.meta}
\usetikzlibrary{intersections}
\usepackage{pgfplots}
\pgfplotsset{compat=1.18}
\usepgfplotslibrary{fillbetween}
 
\tikzset{patient/.style={circle,minimum size=7mm,inner sep=0pt,draw,semithick}}
\tikzset{donor/.style={minimum size=7mm,inner sep=0pt,draw,semithick}}
\tikzset{outdonor/.style={minimum size=7mm, double, inner sep=0pt,draw, semithick}}
\tikzset{indonor/.style={minimum size=7mm, inner sep=0pt,draw,ultra thick}}
\tikzset{kidneyarcs/.style={semithick, arrows = {-Stealth}}}
\tikzset{boldkidneyarcs/.style={ultra thick, blue, arrows = {-Stealth}}}

\usepackage{todonotes}

\newcommand{\NN}{\ensuremath{\mathbb{N}}}

\newcommand{\DDKE}[2]{\ensuremath{(#1,#2)\mathrm{-DDKE}}}
\newcommand{\OPT}{\ensuremath{\mathrm{OPT}}}

\title{\bf Kidney Exchange with Dual Donors}

\author{Ioannis Caragiannis\\Aarhus University\\\url{iannis@cs.au.dk} \and Tomasz Filip Gawron\\Aarhus University\\\url{tomasz.gawron.901@gmail.com} \and Kristoffer Arnsfelt Hansen\\Aarhus University\\\url{arnsfelt@cs.au.dk} \and Yuchen Zhang\\Aarhus University\\\url{deemozyc@gmail.com} }

\date{}

\allowdisplaybreaks

\begin{document}

\maketitle

\begin{abstract}
\noindent Kidney exchange programs enable patients with incompatible living donors to receive transplants by exchanging donors with other patient-donor pairs. Existing models typically assume either that each patient enrolls with a single donor, or that at most one of a patient's multiple alternative donors is selected to donate. In practice, however, a patient may have access to multiple willing but incompatible donors who are all prepared to donate simultaneously if doing so enables the patient to receive a kidney.

We study a dual-donor kidney exchange model, in which a patient may enroll with up to two proxy donors and the exchange is allowed to utilize both, provided the patient receives a compatible transplant. This paradigm allows a patient to participate in an exchange cycle via one donor while their second donor initiates a donation chain—unlocking additional transplants without relying on altruistic donors. We formulate the resulting combinatorial optimization problems under realistic cycle and chain constraints, showing that utilizing both donors introduces sharp computational challenges compared to standard exchange, and we design new approximation algorithms with provable guarantees. We also study strategic manipulation when patients hide a donor, establishing the limits of deterministic mechanisms and showing how randomization achieves truthfulness with near-optimal welfare. Finally, our experiments on realistic synthetic graphs demonstrate that using both donors yields substantial gains in total transplants over selecting at most one—even when only a small fraction of patients enroll with two donors.

\end{abstract}

\section{Introduction}

Kidney transplantation is widely regarded as the most effective treatment for patients with end-stage renal disease~\cite{wolfe1999comparison}. However, due to the severe shortage of deceased-donor organs, living kidney donation has become increasingly important~\cite{biro2019building}. In many cases, a willing living donor is biologically incompatible with their intended recipient due to ABO blood-type mismatch or human leukocyte antigen (HLA) tissue cross-match incompatibility~\cite{rees2009nonsimultaneous, biro2021modelling}. To address this inefficiency, kidney exchange (KE) programs allow incompatible patient--donor pairs to exchange donors with other pairs, enabling mutually compatible transplants~\cite{roth2004kidney, roth2005pairwise}.

The standard theoretical model of kidney exchange represents patient--donor pairs as nodes in a directed compatibility graph, where a directed edge indicates that the donor of one pair is compatible with the patient of another. Directed cycles correspond to feasible exchanges in which each patient receives a kidney. This framework was extended to incorporate non-directed (altruistic) donors, giving rise to chains~\cite{roth2006utilizing, rees2009nonsimultaneous}, and to algorithmic formulations that maximize the number of transplants subject to operational constraints~\cite{abraham2007clearing}. These models have had substantial practical impact: large-scale kidney exchange programs now operate across the US, the UK, and Europe, and periodically clear exchanges using integer programming and related optimization techniques~\cite{manlove2012paired, constantino2013kidney, biro2019building, ashlagi2021operations}. Due to logistical and ethical constraints, real-world programs restrict cycle lengths (typically to at most two or three), since all transplants within a cycle must be performed simultaneously to avoid the risk of donor withdrawal, and similarly bound chain lengths to limit failure propagation~\cite{biro2021modelling}. These restrictions give rise to challenging combinatorial optimization problems, many of which are NP-hard and even APX-hard~\cite{abraham2007clearing, biro2009maximum, constantino2013kidney, ashlagi2021operations}.

Beyond efficiency, kidney exchange has become a canonical setting for studying incentives and market design in matching markets without money. Because the compatibility graph relies on private information reported by participants, a rich literature investigates strategic behavior by patients and hospitals, trade-offs between efficiency and truthfulness, and the design of mechanisms that discourage misreporting or strategic donor withholding~\cite{roth2005pairwise, ashlagi2014free, hatfield2012strategyproof, feigenbaum2021unbalanced}.

\subsection{Dual-donor kidney exchange}

The classical kidney exchange model assumes that each patient is associated with a single donor. While analytically convenient, this assumption is increasingly misaligned with practice, where a patient may have access to multiple willing but incompatible donors, such as multiple family members. Existing models partially accommodate such situations via \emph{alternative donors}, where a patient reports multiple donors to increase compatibility options, but at most one donor is selected in the exchange~\cite{constantino2013kidney,ashlagi2021operations}. While this preserves the structure of the standard kidney exchange problem, it fails to exploit the full potential of multi-donor participation.

A more expressive possibility is to allow multiple donors per patient to donate simultaneously, provided that their associated patient receives a kidney---a paradigm also referred to as \emph{unbalanced matching}~\cite{feigenbaum2021unbalanced, feigenbaum2022strategyproofness}. This many-for-one exchange paradigm allows a patient--donor group to contribute an additional kidney to the pool in return for being served. Crucially, a patient may participate in an exchange cycle via one donor while their second donor initiates a donation chain---effectively acting as an endogenous altruistic donor even in the absence of non-directed donors. This additional flexibility can unlock exchanges that would otherwise be infeasible, benefiting both the multi-donor patient (by making them more attractive to match) and the broader patient population. At the same time, it fundamentally alters the structure of kidney exchange: the problem is no longer a pure cycle-packing problem, and while standard formulations~\cite{abraham2007clearing} remain applicable, they fail to exploit the full potential of the additional proxy donors. 

In this paper, we formulate and study the \emph{dual-donor kidney exchange (DDKE)} model, in which each patient may enroll with up to two proxy donors, and the exchange mechanism may utilize both donors subject to the constraint that a proxy donor donates only if their associated patient receives a compatible kidney. Focusing on at most two donors per patient captures the most realistic multi-donor scenario in practice, while making our computational hardness and strategic impossibility results as strong as possible. Our model strictly generalizes both the classical patient--donor pair model and the alternative-donors model.

We investigate the DDKE model from computational, algorithmic, strategic, and empirical perspectives. Formulating the optimization problem as $(k,\ell)$-DDKE---parameterized by maximum cycle length $k$ and chain length $\ell$---we observe that while the chain-free case ($\ell=0$) reduces to standard kidney exchange with alternative donors (which is solvable in polynomial time for $k=2$ and $k=\infty$), allowing chains renders all natural variants computationally hard: extending even the tractable $2$-cycle or unconstrained-cycle settings to include chains of length just $\ell=1$ makes $(k,1)$-DDKE APX-hard for all $k \ge 2$, while unbounded chains ($\ell=\infty$) lead to strong inapproximability results. For the practically important $(2,1)$-DDKE variant, where ignoring chains yields a simple $1/2$-approximation, we leverage $k$-set packing to design a non-trivial polynomial-time algorithm with an approximation ratio approaching $36/67 \approx 0.537$. From a strategic perspective, patients have a natural incentive to hide a second donor if one suffices to secure a transplant; we prove that no deterministic truthful mechanism can achieve an approximation ratio better than $1/2$, whereas randomization restores truthfulness in expectation while approaching optimal social welfare (i.e., serving the maximum number of patients). Finally, experiments on realistic synthetic compatibility graphs comparing $(2,1)$-DDKE against the alternative-donor baseline $(2,0)$-DDKE demonstrate that utilizing both donors yields substantial gains in total transplants, even when only a small fraction of patients enroll with two donors.

\subsection{Related work and comparison}

The foundations of kidney exchange were laid by Roth, S\"onmez, and \"Unver~\cite{roth2004kidney, roth2005pairwise}, who modeled exchange using cycles among incompatible patient--donor pairs and established core theoretical properties. Subsequent work incorporated algorithmic perspectives, focusing on maximizing the number of transplants under operational constraints. Abraham, Blum, and Sandholm~\cite{abraham2007clearing} introduced graph-based formulations and showed that while exchanges restricted to $2$-cycles (or with unbounded cycle length) admit polynomial-time algorithms, restricting exchanges to cycles of bounded length $k \ge 3$ renders the problem NP-complete. A substantial literature has since studied algorithmic approaches to kidney exchange, including integer programming formulations, heuristic and exact solvers, and extensions incorporating altruistic donors and chains~\cite{constantino2013kidney, anderson2015finding, ashlagi2021operations}.

From a complexity-theoretic perspective, kidney exchange problems exhibit sharp tractability boundaries. While $2$-cycle and unconstrained-cycle exchanges are solvable in polynomial time, bounding cycle or chain lengths quickly leads to NP-hardness~\cite{abraham2007clearing} and APX-completeness~\cite{biro2009maximum}, with inapproximability bounds established via reductions from $k$-set packing~\cite{luo2016approximation}. These hardness results have motivated approximation algorithms based on set packing~\cite{biro2009maximum} as well as practical heuristics for large-scale pools. In contrast to classical exchange where $k=2$ and $k=\infty$ are tractable, we show that in our dual-donor setting even chains of length $\ell = 1$ induce APX-hardness across all cycle bounds $k \ge 2$, and we show how to reduce bounded dual-donor structures to $k$-set packing~\cite{cygan2013improved} to improve upon the baseline $1/2$-approximation.

Kidney exchange has also served as a central application domain for mechanism design without money. Ashlagi and Roth~\cite{ashlagi2014free} and Hatfield~\cite{hatfield2012strategyproof} studied incentives, showing that worst-case efficiency and strategyproofness are incompatible when transplant centers can strategically withhold internal pairs. This impossibility motivated the study of approximate mechanism design in multi-hospital kidney exchange, characterizing the best approximation ratios achievable by deterministic and randomized truthful mechanisms~\cite{ashlagi2015mix, caragiannis2015improved}. At the individual participant level, Roth et al.~\cite{roth2005pairwise} showed that in classical pairwise exchange where at most one alternative donor per patient is used, priority mechanisms are donor-monotonic, giving patients no incentive to hide backup donors. In our dual-donor setting, we show that while lexicographic welfare maximization already prevents an unmatched patient from becoming matched by hiding a donor (\emph{weak truthfulness}), a matched patient still has a strong incentive to hide a second donor if one donor suffices to secure a kidney, leading to a sharp separation between deterministic and randomized truthful mechanisms.

Several papers consider generalizations of kidney exchange in which participants contribute multiple donors or belong to larger groups. Farina, Dickerson, and Sandholm~\cite{farina2017operation} introduced clubs and operation frames, providing integer programming formulations where donors donate on behalf of sets of patients. Ergin, S\"onmez, and \"Unver~\cite{ergin2017dual} studied dual-donor organ exchange from an economic perspective; however, their model addresses dual-graft transplants (such as liver or lung lobes) where a single patient requires simultaneous donations from two donors, whereas in our model each patient requires a single kidney and their two proxy donors can supply kidneys to two distinct recipients. Closest to our setting, Feigenbaum~\cite{feigenbaum2021unbalanced, feigenbaum2022strategyproofness} formulated the concept of \emph{unbalanced matchings} for \emph{overloaded patients} who enroll with multiple proxy donors, focusing primarily on strategic behavior under random transplant cancellations. Building on this direction, we parameterize the failure-free dual-donor problem by cycle and chain lengths $(k,\ell)$ and establish tight worst-case complexity, approximation, and truthfulness guarantees.

Empirical and synthetic evaluations play an important role in assessing the practical relevance of kidney exchange models. Building on the clinical simulation distributions of Saidman et al.~\cite{saidman2006increasing}, Toulis and Parkes~\cite{toulis2011random} formalized a random graph model capturing blood-type and tissue compatibility that has been widely adopted in experimental studies. Our experimental evaluation extends their framework to incorporate dual proxy donors.

\section{Preliminaries}\label{sec:prelim}
We consider extensions of the standard setting where each patient enrolls in the kidney exchange program together with an intended but incompatible donor. In our case, each patient may enroll in the program with either none, one, or two intended but incompatible donors. We represent the input by a directed graph (see Figure~\ref{fig:input-1}). The graph has two sets of nodes: the patient nodes and the donor nodes, illustrated with circles and squares, respectively. The directed edges belong to two categories as well. An edge from a patient node to a donor node indicates a patient who brings a donor to the system; we will refer to such a donor as a proxy donor of that patient. An edge from a donor node to a patient node indicates that a donor is compatible (and can donate their kidney) to the patient. In the example of Figure~\ref{fig:input-1}, there are six patients $p_1, p_2, \dots, p_6$. Patients $p_1$, $p_2$, $p_3$ and $p_5$ have one proxy donor each ($d_1$, $d_2$, $d_3$, and $d_5$, respectively). Patient $p_4$ has two proxy donors, namely donors $p_{4,1}$ and $p_{4,2}$. Patient $p_6$ has no proxy donor. The directed edges from donors to patients indicate all compatibilities. For example, donor $d_1$ is compatible with patients $p_2$ and $p_4$. 

\tikzset{
  pics/combined-example-1/.style={
    code={
        \node[patient] (p1) at ( -1, -1 ) {$p_1$};
        \node[patient] (p2) at ( -3, 1 ) {$p_2$};
        \node[patient] (p3) at ( -5, -1 ) {$p_3$};
        \node[patient] (p4) at ( 1, 1 ) {$p_4$};
        \node[patient] (p5) at ( 3, -1) {$p_5$};
        \node[patient] (p6) at ( 5, 1) {$p_6$};

        \node[donor] (d1) at (-1, 1) {$d_1$};
        \node[donor] (d2) at (-5, 1) {$d_2$};
        \node[donor] (d3) at (-3, -1) {$d_3$};
        \node[donor] (d4-1) at (1, -1) {$d_{4,1}$};
        \node[donor] (d4-2) at (3, 1) {$d_{4,2}$};
        \node[donor] (d5) at (5, -1) {$d_5$};

        \draw[kidneyarcs]
            (p1) edge (d1) 
            (d1) edge (p2) 
            (p2) edge (d2) 
            (d2) edge (p3) 
            (p3) edge (d3) 
            (d3) edge (p1)
            (d1) edge (p4) 
            (p4) edge (d4-1) edge (d4-2)
            (d4-1) edge (p1) 
            (d4-2) edge (p5)
            (p5) edge (d5) 
            (d5) edge (p6) 
        ;
        }
    }
}

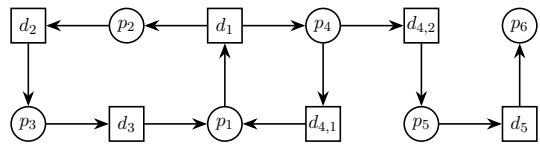
\begin{figure}[ht]
\centering
\begin{tikzpicture}[scale=0.65, transform shape]
  \pic {combined-example-1};
\end{tikzpicture}
\caption{An example of an input graph in our dual donor kidney exchange setting.}
\label{fig:input-1}
\end{figure}

We may follow the traditional approach in kidney exchange and compute the cycle of transplantations of maximum length. In our example, such a cycle is formed by nodes $p_1$, $d_1$, $p_2$, $d_2$, $p_3$, $d_3$, and $p_1$ and the directed edges between consecutive nodes in this order (see~Figure~\ref{fig:two-solutions-1}, left), implying the following three transplantations: the proxy donor $d_1$ of patient $p_1$ donates their kidney to patient $p_2$, the proxy donor $d_2$ of patient $p_2$ donates their kidney to patient $p_3$, and the proxy donor $d_3$ of patient $p_3$ donates their kidney to patient $p_1$. We will refer to such a cycle as a cycle of length three (due to the fact that three patients are served with a kidney).

However, there is another set of transplantations that serves four patients. This is possible by exploiting the fact that patient $p_4$ has two proxy donors that can be used in transplantations, provided that patient $p_4$ is served with a kidney. So, in a smaller cycle that serves two patients, the proxy donor $d_1$ of patient $p_1$ donates to patient $p_4$, whose first proxy donor $d_{4,1}$ donates to patient $p_1$. The second proxy donor $d_{4,2}$ of patient $p_4$ donates their kidney to patient $p_5$, whose proxy donor $d_5$ donates to patient $p_6$. This set of transplantations is depicted at the right of Figure~\ref{fig:two-solutions-1}. Essentially, the patient $p_4$ uses their first proxy donor $d_{4,1}$ to form a cycle of transplantations and their second proxy donor $d_{4,2}$ to initiate a chain of transplantations. This example suggests that simply covering the maximum number of patients in cycles may not be the optimal solution in our setting. 

\begin{figure}[ht]
\centering
\begin{subfigure}{0.45\textwidth}
\centering
\begin{tikzpicture}[scale=0.65, transform shape]
  \pic {combined-example-1};
  \draw[boldkidneyarcs]
    (p1) edge (d1)
    (d1) edge (p2)
    (p2) edge (d2)
    (d2) edge (p3)
    (p3) edge (d3)
    (d3) edge (p1)
    ;
\end{tikzpicture}
\end{subfigure}
\hspace{5mm}
\begin{subfigure}{0.45\textwidth}
\centering
\begin{tikzpicture}[scale=0.65, transform shape]
  \pic {combined-example-1};
  \draw[boldkidneyarcs]
    (p1) edge (d1)
    (d1) edge (p4)
    (p4) edge (d4-1) edge (d4-2)
    (d4-1) edge (p1)
    (d4-2) edge (p5)
    (p5) edge (d5)
    (d5) edge (p6)
    ;
\end{tikzpicture}
\end{subfigure}
\caption{Serving the longest cycle containing three patients (left) or a pair of patients in a cycle and a chain containing two more patients (right) for the input graph of Figure~\ref{fig:input-1}.}
\label{fig:two-solutions-1}
\end{figure}
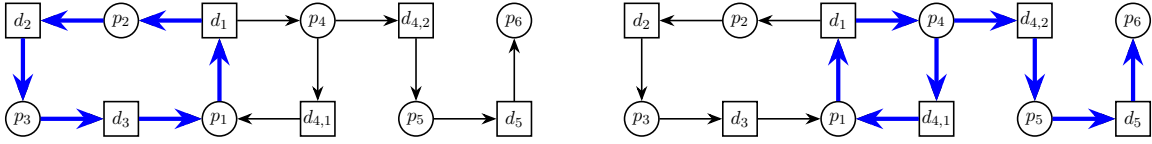

The input graph of Figure~\ref{fig:input-1} contains cycles of different lengths. However, we can make a similar observation in the case that either the input graph does not contain cycles of length higher than two, or we are not allowed to form transplantation cycles of length higher than two. Indeed, consider the next input graph depicted in the two solutions of Figure~\ref{fig:two-solutions-2}. There are three $2$-cycles that overlap. The left part of Figure~\ref{fig:two-solutions-2} depicts the solution consisting of the two disjoint cycles serving the four patients $p_1$, $p_2$, $p_3$, and $p_4$. On the right part of Figure~\ref{fig:two-solutions-2}, we can see a solution that exploits the fact that patients $p_1$ and $p_2$ have two proxy donors each. The solution consists of the $2$-cycle including patients $p_1$ and $p_2$ (and donors $d_{1,2}$ and $d_{2,2}$) and two chains that originate from the other proxy donors $d_{1,1}$ and $d_{2,1}$ of these patients. This solution serves six patients overall and is optimal. It can be verified that no other solution can serve more than five patients.

\tikzset{
  pics/combined-example-2/.style={
    code={
        \node[patient] (p1) at ( 0, 2 ) {$p_1$};
        \node[patient] (p2) at ( -2, 0 ) {$p_2$};
        \node[patient] (p3) at ( 2, 4 ) {$p_3$};
        \node[patient] (p4) at ( -4, -2 ) {$p_4$};
        \node[patient] (p5) at ( 2, 0) {$p_5$};
        \node[patient] (p6) at ( 0, -2) {$p_6$};
        \node[patient] (p7) at ( 4, 2) {$p_7$};
        \node[patient] (p8) at ( -2, -4) {$p_8$};

        \node[donor] (d1-1) at (2, 2) {$d_{1,1}$};
        \node[donor] (d1-2) at (0, 0) {$d_{1,2}$};
        \node[donor] (d2-1) at (-2, -2) {$d_{2,1}$};
        \node[donor] (d2-2) at (-2, 2) {$d_{2,2}$};
        \node[donor] (d3) at (0, 4) {$d_3$};
        \node[donor] (d4) at (-4, 0) {$d_4$};
        \node[donor] (d5) at (4, 0) {$d_5$};
        \node[donor] (d6) at (0, -4) {$d_6$};

        \draw[kidneyarcs]
            (p1) edge (d1-1) edge (d1-2) 
            (p2) edge (d2-1) edge (d2-2) 
            (p3) edge (d3)
            (p4) edge (d4)
            (p5) edge (d5)
            (p6) edge (d6)

            (d1-1) edge (p3) edge (p5)
            (d1-2) edge (p2)
            (d2-1) edge (p4) edge (p6)
            (d2-2) edge (p1)
            (d3) edge (p1)
            (d4) edge (p2)
            (d5) edge (p7)
            (d6) edge (p8)
        ;
        }
    }
}

\begin{figure}[ht]
\centering
\begin{subfigure}{0.45\columnwidth}
\centering
\begin{tikzpicture}[scale=0.65, transform shape]
  \pic {combined-example-2};
  \draw[boldkidneyarcs]
    (p1) edge (d1-1)
    (d1-1) edge (p3)
    (p3) edge (d3)
    (d3) edge (p1)
    (p2) edge (d2-1)
    (d2-1) edge (p4)
    (p4) edge (d4)
    (d4) edge (p2)
    ;
\end{tikzpicture}
\end{subfigure}
\begin{subfigure}{0.45\columnwidth}
\centering
\begin{tikzpicture}[scale=0.65, transform shape]
  \pic {combined-example-2};
  \draw[boldkidneyarcs]
    (p1) edge (d1-2) edge (d1-1)
    (d1-2) edge (p2)
    (p2) edge (d2-1) edge (d2-2)
    (d2-2) edge (p1)
    (d1-1) edge (p5)
    (d2-1) edge (p6)
    (p5) edge (d5)
    (p6) edge (d6)
    (d5) edge (p7)
    (d6) edge (p8)
    ;
\end{tikzpicture}
\end{subfigure}
\caption{Serving the maximum number of two $2$-cycles (left) and a single $2$-cycle that allows to serve two long chains (right).}\label{fig:two-solutions-2}
\end{figure}
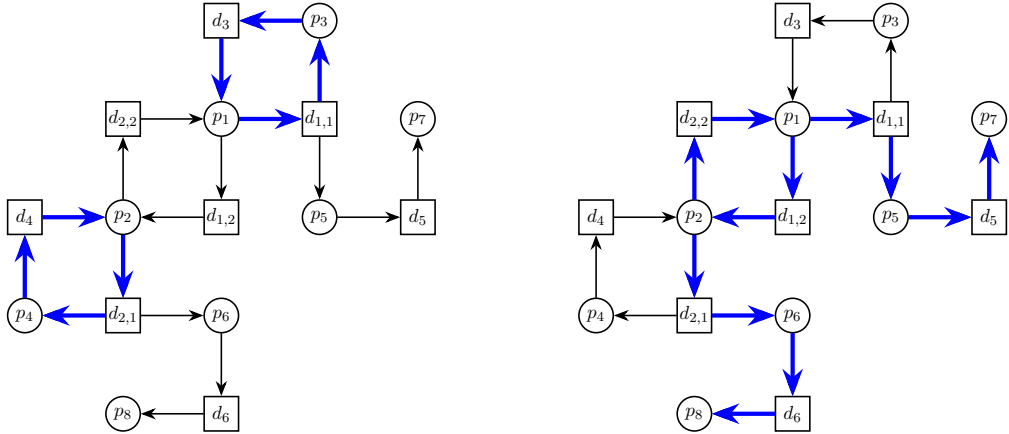

\section{Problem variants and their complexity}\label{sec:complexity}
In this section, we define formally variants of the problem we consider in the rest of the paper and present corresponding complexity results. 
We begin with the definition of the input.

\begin{definition}[Kidney exchange graph]
A \emph{kidney exchange graph} is a directed graph $G=(V,E)$, where the set of nodes $V$ is the disjoint union of a set of patients $V_P$ and a set of donors $V_D$. The directed edges either go from a patient node~$p$ to a donor node~$d$, in which case~$d$ is a proxy donor for $p$, or from a donor node~$d$ to a patient node~$p$, in which case~$d$ is donor compatible with~$p$. The outdegree of a patient node~$p$ is thus the number of proxy donors enrolled with~$p$. A donor node has indegree at most~1, meaning that~$d$ is either an altruistic donor or enrolls with a specific patient.
\end{definition}

In the model we study each patient enrolls with at most~$2$ proxy donors and no donors are altruistic. Altruistic donors can however be \emph{simulated} in our model by a 2-cycle of two patients~$p_1$ and~$p_2$ and their corresponding proxy donors~$d_1$ and~$d_{2,1}$ and then giving $p_2$ an additional proxy donor $d_{2,2}$, see Figure~\ref{fig:altruistic-donor-gadget}, where the simulated altruistic donor node is identified using a double-lined square.

\begin{figure}[ht]
\centering
\begin{tikzpicture}[scale=0.65, transform shape]
    \node[patient] (p1) at (0,0) {$p_1$};
    \node[patient] (p2) at (2,2) {$p_2$};
    \node[donor] (d1) at (2,0) {$d_1$};
    \node[donor] (d2-1) at (0,2) {$d_{2,1}$};
    \node[outdonor] (d2-2) at (4,2) {$d_{2,2}$};    

    \draw[kidneyarcs] 
    (p1) edge (d1)
    (d1) edge (p2)
    (p2) edge (d2-1) edge (d2-2)
    (d2-1) edge (p1)
    ;
\end{tikzpicture}
\caption{Altruistic donor gadget.}
\label{fig:altruistic-donor-gadget}
\end{figure}
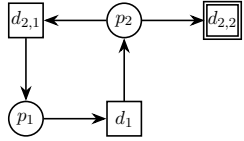

We are interested in \emph{cycles} and \emph{chains} of kidney exchanges.
\begin{definition}
Given a kidney exchange graph $G$, a donor \emph{cycle} is simply a directed cycle in~$G$ and a donor \emph{chain} is a directed path connecting a donor node to a patient node. The \emph{length} of a donor cycle or donor chain is the number of patient nodes they contain.
\end{definition}
Cycles and chains form the basic building blocks of kidney exchange allocations.\footnote{We often use the term ``solution'' as an alternative to ``allocation''.}
\begin{definition}[Kidney exchange allocations]
A kidney exchange \emph{allocation} for a given kidney exchange graph $G$ is a subgraph $M$ of $G$ such that: 
\begin{enumerate*}[label=(\roman*)]\item every patient node $p$ in $M$ has indegree~$1$ and \item every donor node $d$ in $M$ has outdegree~$1$ and has the same indegree as in $G$.
\end{enumerate*}
The \emph{cycle length} of an allocation $M$ is the length of the longest donor cycle in the allocation. The \emph{chain length} of an allocation is the length of the longest donor chain in the allocation obained from $M$ by removing all donor cycles.
\end{definition}
The basic computational problem in a kidney exchange graph is to compute an allocation containing the largest possible number of patients. We are interested in the problem where each patient have~0, 1, or~2~proxy donors and there are no altruistic donors. We parameterize this problem in terms of cycle and chain length.
\begin{definition}[Dual donor kidney exhange problem]
Let $k,\ell \in \NN \cup \{\infty\}$. The $\DDKE{k}{\ell}$ problem is, given a kidney exchange graph where the outdegree of each patient is at most~$2$ and the indegree of every donor node is~$1$, to compute a kidney exchange allocation of cycle length at most~$k$ and chain length at most~$\ell$ such that the number of patient nodes in the allocation is maximized.
\end{definition}
\begin{remark}
    The $\DDKE{k}{0}$ problem may be viewed simply as the classical paired kidney exchange problem with exchange cycles of length at most~$k$. Since no chains are possible, patient nodes of outdegree~0 can safely be removed and for patient nodes with outdegree~2 we may simply merge the associated two donor nodes. This observation settles the computational complexity of $\DDKE{k}{0}$, since it can be solved in polynomial time for $k=2$ and $k=\infty$~\cite{abraham2007clearing,biro2009maximum}, whereas for every fixed $3\leq K < \infty$, it is APX-complete~\cite{biro2009maximum}. For $k=3$, specifically, Luo~et~al.~\cite{luo2016approximation} proved that the problem is NP-hard to approximatete within~$\frac{698}{697}-\eps$, for any constant $\eps>0$.
\end{remark}
In the remainder of this section we give hardness results for the cases where $\ell>0$. These results are obtained by reductions built from kidney exchange graph \emph{gadgets}. 
Even though we ultimately don't allow altruistic donors, it is useful to have donor nodes of 
indegree~$0$ serving as \emph{input nodes}. Other donor nodes will serve as \emph{output nodes}, and gadgets are combined by identifying input and output nodes of the respective gadgets. In this way a gadget is only connected to the rest of the graph via directed edges  to inputs and or from outputs. The same node never serves as both an input node and an output node.

In figures we illustrate input nodes by thick-lined squares and output nodes by double-lined squares.
\begin{definition}[Used and provided donors]
   Let $G=(V,E)$ be a kidney exchange graph with a set $I \subseteq V_D$ of \emph{input nodes} each of indegree~$0$, a set $O \subseteq V_D$ of \emph{output nodes}, and consider an allocation $M$ in $G$. When $d \in I \cap M$ we say that $M$ \emph{uses} the input node $d$. When $p \in V_P \cap M$, $d \in O \setminus M$, and $(p,d) \in E$ we say that $M$ \emph{provides} the output node $d$.
\end{definition}

When gadgets are combined to form a larger kidney exchange graph there is a natural relationship between allocations in the gadgets and in the combined graph.
\begin{definition}
Let $G_1=(V_1,E_1)$ and $G_2=(V_2,E_2)$ be kidney exchange graphs such that $V_1 \cap V_2 = I \cup O$, where the nodes in $I$ are input nodes of $G_1$ and output nodes of $G_2$ and the nodes in $O$ are output nodes of $G_1$ and input nodes of $G_2$.

An allocation $M$ in the combined graph $G_1 \cup G_2$ naturally \emph{induces} allocations $M_1$ and $M_2$ in $G_1$ and $G_2$ as follows. Let $I'$ be the donor nodes in $I$ with a directed edge to a patient node of $G_1$ and let $O'$ be the donor nodes in $O$ with a directed edge to a patient node of $G_2$. Then $M_1=(G_1 \cap M) \setminus O'$ is an allocation in $G_1$ using $I'$ and providing $O'$ and  $M_2=(G_2 \cap M) \setminus I'$ is an allocation in $G_2$ using $O'$ and providing $I'$.

Converserly, if $M_1$ is an allocation of $G_1$ using $I' \subseteq I$ and providing $O' \subseteq O$ and $M_2$ is an allocation of $G_2$ using $O'$ and providing $I'$, the (unique) allocation $M$ of $G_1 \cup G_2$ that induces $M_1$ and $M_2$ is the \emph{combined} allocation of $M_1$ and $M_2$.
\end{definition}

Our first result is that the $\DDKE{k}{1}$ is APX-hard for every $2 \leq k \leq \infty$ and the result is based on the gadget illustrated in Figure~\ref{fig:combined-choice-duplication-gadget} and whose maximal allocations are analyzed below. 

\tikzset{
  pics/combined-gadget-base/.style={
    code={
        \node[patient] (p1) at ( 0, 1.5) {$p_1$};
        \node[patient] (p2) at ( 0,-1.5) {$p_2$};
        \node[patient] (p3) at (-4, 0) {$p_3$};
        \node[patient] (p4) at ( 4, 0) {$p_4$};
   
        \node[outdonor] (d1-1) at (-2.5, 1.5) {$d_{1,1}$};
        \node[outdonor] (d2-2) at (-2.5,-1.5) {$d_{2,2}$};
        \node[outdonor] (d2-1) at (2.5,-1.5) {$d_{2,1}$};
        \node[outdonor] (d1-2) at (2.5, 1.5) {$d_{1,2}$};

        \node[donor] (d3) at (-1.5, 0) {$d_3$};
        \node[donor] (d4) at ( 1.5, 0) {$d_4$};
    }
  }
}

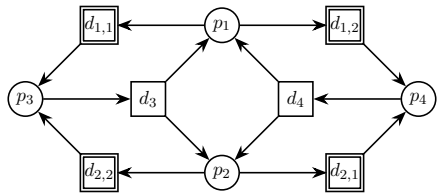
\begin{figure}[ht]
\centering
\begin{tikzpicture}[scale=0.65, transform shape]
  \pic {combined-gadget-base};
  \draw[kidneyarcs]
      (p1) edge (d1-1) edge (d1-2)
      (p2) edge (d2-1) edge (d2-2)
      (p3) edge (d3)
      (p4) edge (d4)
      (d1-1) edge (p3)
      (d1-2) edge (p4)
      (d2-1) edge (p4)
      (d2-2) edge (p3)
      (d3) edge (p1) edge (p2)
      (d4) edge (p1) edge (p2)
      ;
\end{tikzpicture}
\caption{Combined choice and duplication gadget.}
\label{fig:combined-choice-duplication-gadget}
\end{figure}

\begin{lemma}[Combined choice and duplication gadget]
    Let $G$ be the kidney exchange graph given in Figure~\ref{fig:combined-choice-duplication-gadget} having no input nodes, 4~patient nodes $p_1$,$p_2$,$p_3$,$p_4$, and 4~output donor nodes $d_{1,1}$,$d_{1,2}$,$d_{2,1}$,$d_{2,2}$. We have the following.
    \begin{enumerate}
    \item There are~2 allocations of cycle length~$2$ and no chains, that contains all patient nodes and each of these providing~2 output nodes. One provides $d_{1,1}$ and $d_{2,1}$ and the other provides $d_{1,2}$ and $d_{2,2}$, as shown in Figure~\ref{fig:combined-choice-duplication-gadget-allocations-short-d1-1-d2-1} and Figure~\ref{fig:combined-choice-duplication-gadget-allocations-short-d1-2-d2-2}.
    \item No allocation of chain length at most~1 provides nodes from both of the sets $\{d_{1,1}$ and $d_{2,1}\}$ and $\{d_{1,2}$ and $d_{2,2}\}$.
    \end{enumerate}
    \label{lem:gadget-analysis}
\end{lemma}
\begin{proof}
    Consider an allocation $M=(V_M,E_M)$ in $G$ with at least one cycle. By symmetry, we may analyze just the case of $p_1 \in V_M$ being in a cycle. Then we may, again by symmetry, analyze just the case of $(p_1,d_{1,1}) \in E_M$ being an edge of the cycle.   
    To complete the cycle, this necessitates that it continues through $p_3$ and $d_3$. The cycle may now be closed in precisely two ways, either directly to $p_1$ giving a 2-cycle, or through intermediate vertices $p_2$, $d_{2,1}$, $p_4$, and $d_4$ before returning to $p_1$, giving a 4-cycle.

    The 4-cycle uses all patient nodes, leaving no room for further cycles or chains, and provides the output nodes $d_{1,2}$ and $d_{2,2}$ as shown in Figure~\ref{fig:combined-choice-duplication-gadget-allocations-long-d1-2-d2-2}.

    The 2-cycle allows for a second 2-cycle which together provides the output nodes $d_{1,2}$ and $d_{2,2}$ as illustrated in Figure~\ref{fig:combined-choice-duplication-gadget-allocations-short-d1-2-d2-2}, leaving no room for additional cycles or chains. The 2-cycle also allows for a chain starting from $d_{1,2}$ followed by $p_4$. Ending the chain in $p_4$ provides $d_4$, but none of the outputs, and serves only~3 patients. The only alternative is to continue the chain through $d_4$ and ending in $p_2$, thereby serving all~4 patients, but doing so gives a chain of length~2.
\end{proof}

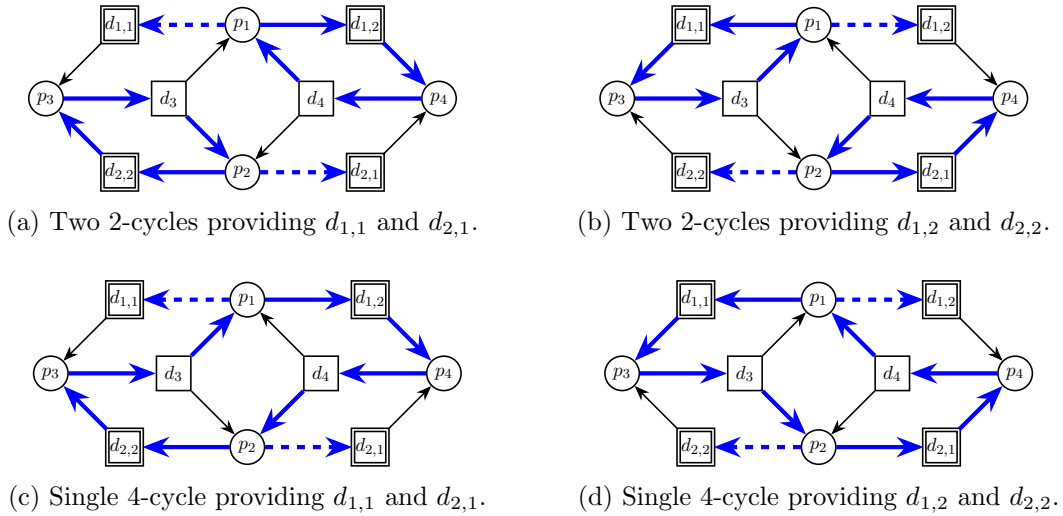
\begin{figure}[ht]
\centering
\begin{subfigure}{0.45\textwidth}
\centering
\begin{tikzpicture}[scale=0.65, transform shape]
  \pic {combined-gadget-base};
  \draw[kidneyarcs]
    (d1-1) edge (p3)
    (d3) edge (p1)
    (d2-1) edge (p4)
    (d4) edge (p2)
    ;
    
  \draw[boldkidneyarcs]
    (p1) edge[dashed] (d1-1)
    (p1) edge  (d1-2)
    (d1-2) edge (p4)
    (p4) edge (d4)
    (d4) edge (p1)

    (p2) edge[dashed] (d2-1)
    (p2) edge (d2-2)
    (d2-2) edge (p3)
    (p3) edge (d3)
    (d3) edge (p2)
    ;
\end{tikzpicture}
\caption{Two 2-cycles providing $d_{1,1}$ and $d_{2,1}$.}
\label{fig:combined-choice-duplication-gadget-allocations-short-d1-1-d2-1}
\end{subfigure}
\begin{subfigure}{0.45\textwidth}
\centering
\begin{tikzpicture}[scale=0.65, transform shape]
  \pic {combined-gadget-base};
  \draw[kidneyarcs]
    (d1-2) edge (p4)
    (d4) edge (p1)
    (d2-2) edge (p3)
    (d3) edge (p2)
    ;
    
  \draw[boldkidneyarcs]
    (p1) edge  (d1-1)
    (p1) edge[dashed] (d1-2)
    (d1-1) edge (p3)
    (p3) edge (d3)
    (d3) edge (p1)

    (p2) edge (d2-1)
    (p2) edge[dashed] (d2-2)
    (d2-1) edge (p4)
    (p4) edge (d4)
    (d4) edge (p2)
    ;
\end{tikzpicture}
\caption{Two 2-cycles providing $d_{1,2}$ and $d_{2,2}$.}
\label{fig:combined-choice-duplication-gadget-allocations-short-d1-2-d2-2}
\end{subfigure}
\vspace{5mm}

\begin{subfigure}{0.45\textwidth}
\centering
\begin{tikzpicture}[scale=0.65, transform shape]
  \pic {combined-gadget-base};
  \draw[kidneyarcs]
    (d1-1) edge (p3)
    (d3) edge (p2)
    
    (d2-1) edge (p4)
    (d4) edge (p1)
    ;
    
  \draw[boldkidneyarcs]
    (p1) edge[dashed] (d1-1)
    (p1) edge  (d1-2)    
    (d1-2) edge (p4)
    (p4) edge (d4)
    (d4) edge (p2)

    (p2) edge[dashed] (d2-1)
    (p2) edge (d2-2)
    (d2-2) edge (p3)
    (p3) edge (d3)
    (d3) edge (p1)
    ;
\end{tikzpicture}
\caption{Single 4-cycle providing $d_{1,1}$ and $d_{2,1}$.}
\label{fig:combined-choice-duplication-gadget-allocations-long-d1-1-d2-1}
\end{subfigure}
\begin{subfigure}{0.45\textwidth}
\centering
\begin{tikzpicture}[scale=0.65, transform shape]
  \pic {combined-gadget-base};
  \draw[kidneyarcs]
    (d1-2) edge (p4)
    (d4) edge (p2)
    
    (d2-2) edge (p3)
    (d3) edge (p1)
    ;
    
  \draw[boldkidneyarcs]
    (p1) edge  (d1-1)
    (p1) edge[dashed] (d1-2)
    (d1-1) edge (p3)
    (p3) edge (d3)
    (d3) edge (p2)

    (p2) edge (d2-1)
    (p2) edge[dashed] (d2-2)
    (d2-1) edge (p4)
    (p4) edge (d4)
    (d4) edge (p1)
    ;
\end{tikzpicture}
\caption{Single 4-cycle providing $d_{1,2}$ and $d_{2,2}$.}
\label{fig:combined-choice-duplication-gadget-allocations-long-d1-2-d2-2}
\end{subfigure}
\caption{Maximal allocations by cycles.}
\label{fig:combined-choice-duplication-gadget-allocations}
\end{figure}
To obtain hardness for $\DDKE{k}{1}$, we reduce from a bounded occurrence version of the MAXE3SAT problem. An E3SAT formula is a CNF formula in which every clause contains exactly~3 literals. An E3SAT formula is $\delta$-satiable if there exists a truth assignment to the variables such that the fraction of clauses satisfied is at least~$\delta$. Berman et al.~\cite{KPS03} obtained the following hardness result.
\begin{theorem}[Berman, Karpinski and Scott~{\cite[Theorem~2]{KPS03}}]
    For any $\eps>0$ it is NP-hard to distinguish E3SAT formulas, where every literal occurs exactly~2 times, that are at least $(1-\eps)$-satisfiable from those are at most $(1015/1016+\eps)$-satisfiable.
    \label{thm:Berman-Karpinski-Scott}
\end{theorem}
We now present our reduction.
\begin{theorem}\label{thm:reduction}
For any $2 \leq k \leq \infty$ and any $\eps>0$ it is NP-hard to distinguish between $\DDKE{k}{\ell}$ kidney exchange graphs for which there exists an allocation containing at least a fraction $1-\eps/4$ of the patient nodes from those where every allocation contains at most a fraction $4063/4064+\eps/4$ of the patient nodes. Consequently, it is NP-hard to approximate $\DDKE{k}{1}$ within $\frac{4064}{4063}-\eps$, for any $2 \leq k \leq \infty$, and any $\eps>0$.
\label{thm:DDKE-k-1-hardness}
\end{theorem}
\begin{proof}
Given an E3CNF formula $\Phi$ with $n$ variables where every literal occurs exactly~2 times we construct in polynomial time a kidney exchange graph $G$ as follows.

For each variable $x_i$ we create a combined choice and duplication gadget, where the output nodes $d_{1,1}$ and $d_{2,1}$ correspond to the literal $x_i$ and the output nodes $d_{1,2}$ and $d_{2,2}$ correspond to the literal $\neg x_i$. For each clause $C_j$ we create a \emph{clause gadget} that consists of a single patient node $p_j$ and 3~input donor nodes with directed edges to $p_j$ corresponding to the literal of the clause. To combine these to form $G$ we simply identify the output nodes of each combined choice and  duplication gadgets with corresponding input nodes of the patient gadgets, which is possible since each literal occurs exactly~2 times. Since $\Phi$ must have exactly $m=(4/3)n$ clauses the graph $G$ contains exactly $N= 4n+(4/3)n=(16/3)n$ patient nodes.

Consider an allocation $M$ in $G$ with a maximum number of patient nodes. Since each chain in $M$ is of length~1, there can be no chain containing a patient node of a combined choice and duplication gadget. Indeed, that would mean $M$ could contain at most~3 patient nodes of that gadget without providing any output nodes. The allocation $M$ could then be altered locally to contain all~4 patient nodes, contracting maximality. 

Thus $M$ must include one or two cycles through all~4 patients nodes in each combined choice and duplication gadget, and the output nodes not in these cycles correspond exactly to a truth assignment of the variables of $\Phi$ by Lemma~\ref{lem:gadget-analysis}. The patient node in a clause gadget is thus in $M$ if and only if the clause is satisfied by that truth assignment. We can now use this insight to complete the analysis.

Suppose first that $\Phi$ is at least $(1-\eps)$-satisfiable. 
Then there exist an allocation in $G$ with at least at least $4n+(1-\eps)(4/3)n = (1-\eps/4)N$ patient nodes. Suppose next that $\Phi$ is at most $(1015/1016+\eps)$ satisfiable. Then every allocation in $G$ must contain at most $4n+(1015/1016 + \eps)(4/3)n = (4063/4064+\eps/4)N$ patient nodes.
The conclusion then follows from Theorem~\ref{thm:Berman-Karpinski-Scott}.
\end{proof}

\begin{remark}
    The proof of Theorem~\ref{thm:reduction} used the restriction of chains to length~1, but the same reduction in fact also establishes the same hardness result for $\DDKE{k}{2}$. For chains of length~3, however, the reduction no longer work. The reason is that the combined choice and duplication gadget has allocations consisting of a 2-cycle and a 2-chain that provides nodes from both $\{d_{1,1}$ and $d_{2,1}\}$ and $\{d_{1,2}$ and $d_{2,2}\}$.
\end{remark}

To be able to accommodate longer chains we use an alternative combined choice and duplication gadget, given in Figure~\ref{fig:alternative-choice-gadget-and-duplication-gadget}. With that we can show APX-hardness of $\DDKE{k}{\ell}$ for any constant chain length $2\leq \ell < \infty$.
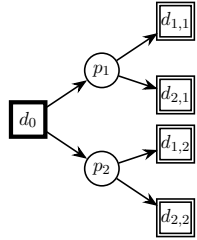
\begin{figure}[ht]
\centering
\begin{tikzpicture}[scale=0.65, transform shape]
    \node[indonor] (d0) at (0,0) {$d_0$};
    \node[patient] (p1) at (1.5,1) {$p_1$};
    \node[patient] (p2) at (1.5,-1) {$p_2$};
    \node[outdonor] (d11) at (3,2) {$d_{1,1}$};
    \node[outdonor] (d21) at (3,0.5) {$d_{2,1}$};
    \node[outdonor] (d12) at (3,-0.5) {$d_{1,2}$};
    \node[outdonor] (d22) at (3,-2) {$d_{2,2}$};

    \draw[kidneyarcs] 
    (d0) edge (p1) edge (p2)
    (p1) edge (d11)
    (p1) edge (d21)
    (p2) edge (d12)
    (p2) edge (d22)
    ;
\end{tikzpicture}
\caption{Alternative combined choice and duplication gadget.}\label{fig:alternative-choice-gadget-and-duplication-gadget}
\end{figure}
\begin{theorem}
    For any $2\leq k \leq \infty$ and any $\eps>0$ it is NP-hard to distinguish between $\DDKE{k}{2}$ kidney exchange graphs for which there exists an allocation containing at least a fraction $1-\eps/4$ of the patient nodes from those where every allocation contains at most a fraction $4063/4064+\eps/4$ of the patient nodes. The graphs produced by the reduction furthermore satisfy that no allocations have chains of length longer than~2. Consequently, it is NP-hard to approximate $\DDKE{k}{\ell}$ within $\frac{4064}{4063} - \eps$, for any $2 \leq k \leq \infty$, any $2 \leq \ell < \infty$, and any $\eps>0$.
\label{thm:DDKE-k-l-hardness}
\end{theorem}
\begin{proof}  
    We modify the reduction of the proof of Theorem~\ref{thm:DDKE-k-1-hardness} by simply replacing the combined choice and consistency gadget. More precisely, for each variable~$x_i$ we take the altruistic donor gadget of Figure~\ref{fig:altruistic-donor-gadget} and identify its output node with the input node of the alternative combined choice and duplication gadget of Figure~\ref{fig:alternative-choice-gadget-and-duplication-gadget}. Note that this introduces exactly~4 patient nodes for every variable which is the same number of patient nodes of the replaced combined and consistency gadget. The remaining part of the proof thus follows exactly as before resulting in a kidney exchange graph~$G$. By construction, the largest possible chain in in allocation in~$G$ starts from the simulated altrustic donor and ends in the patient node of the clause gadget, which is a chain of length~2.
\end{proof}
\begin{remark}
For obtaining stronger inaproximability results we may consider MAXSAT instances with more occurrences of literals and extend the combined choice and duplication gadget accordingly.
Specifically, one may directly extend the combined choice and duplication using a simple duplication gadget (see Figure~\ref{fig:choice-gadget-and-duplication-gadget}) repeatedly to match the number of donor nodes corresponding each literal to its number of occurrences, thereby allowing one to obtain improvements in the precise constant of inapproximability depending on~$\ell$. We leave this to the interested reader.
\end{remark}

Our next results show that, unless P=NP, $\DDKE{k}{\infty}$ is not contained in APX for all $2\leq k \leq \infty$. We obtain the strongest such statement for $k<\infty$ by a simple reduction from the longest path problem in directed graphs.
\begin{theorem}[Björklund, Husfeldt, and~Khanna\cite{bjorklund2004approximation}]
For any $\eps>0$, it is NP-hard to approximate the longest path in a directed graph with~$n$ nodes within $n^{1-\eps}$.
\label{thm:Bjoerklund-Husfeldt-Khanna}
\end{theorem}
The basic idea of the reduction is to model paths by chains, and initiate a chain of kidney exchanges by a simulated altruistic donor.
\begin{theorem}
For every $\eps>0$ and every~$2 \leq k < \infty$ it is NP-hard to approximate $\DDKE{k}{\infty}$ for kidney exchange graphs with~$n$ patients within $n^{1-\eps}$.
\label{thm:DDKE-k-infty-hardness}
\end{theorem}
\begin{proof}
    Given a directed graph $G=(V,E)$ we construct in polynomial time  a kidney exchange graph $G'=(V',E')$ as follows. For each node $v \in V$ we create a patient node $p_v$ with a proxy donor $d_v$. For every directed edge $(u,v) \in E$ we create a path from $d_u$ to $p_v$ alternating between donor nodes and patient nodes through $k-1$ new patient and $k-1$ new donor nodes. In this way $G'$ has no cycles of length less than $2k$, and every path in $G$ from $u$ to $v$ of length $L$ gives rise to length~$kL$ chain from $d_u$ to $p_v$.
    
    In order to initiate the kidney exchanges we use the altruistic donor gadget of Figure~\ref{fig:altruistic-donor-gadget} that provides an output donor. This output donor is then connected by directed edges to all patient nodes~$p_v$ for $v \in V$. The resulting kidney exchange graph $G'$ thus has $N=kn+2$ patient nodes.

    It follows that $G$ has a path of length~$L$ if and only if $G'$ has an allocation with at least $kL+3$ patients, consisting of a 2-cycle and a chain of length $kL+1$. 
     
    Letting $\OPT$ and $\OPT'$ denote length of the longest path in $G$ and the maximum number of patients in an allocation in $G'$, respectively, we then have $k \cdot \OPT+3 \leq \OPT'$. Furthermore, an allocation in $G'$ with more than~2 patients will consist of the 2-cycle and a chain where every $k$-th patient node is of the form $p_u$ for $u \in V$. An allocation in $G'$ with $M$ patients thus gives rise to a path in $G$ of length at least $\left\lfloor (M-(k+1))/k\right\rfloor$. We conclude that an allocation with at least $\OPT' / N^{1-\eps}$ patients gives a path in $G$ of length $\Omega(\OPT/n^{1-\eps})$ and the result follows from Theorem~\ref{thm:Bjoerklund-Husfeldt-Khanna}. 
\end{proof}
For the case when we place no restrictions on cycle or chain length we prove hardness by a more involved reduction from MAXE3SAT involving a hardness amplification construction. We could in principle use the construction of Theorem~\ref{thm:DDKE-k-l-hardness} as the initial source of hardness, but proving a stronger statement results in an overall simpler proof. We use the following celebrated result of Håstad~\cite{H01}.
\begin{theorem}[{Håstad~\cite[Theorem 6.5]{H01}}]
\label{thm:Haastad-E3SAT}
    For any $\eps>0$ it is NP-hard to distinguish satisfiable E3SAT formulas from those that are at most $(7/8+\eps)$ satisfiable.
\end{theorem}
The main advantage of using this result, compared to using~Theorem~\ref{thm:Berman-Karpinski-Scott}, is that one of the two cases involves \emph{satisfiable} formulas. 

Our reduction again uses gadgets for choice and duplication, where we for simplicity split the alternative combined choice and duplication gadget into a separate choice and duplication gadgets as shown in Figure~\ref{fig:choice-gadget-and-duplication-gadget}. 

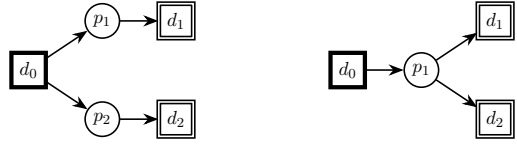
\begin{figure}[ht]
\centering
\begin{tikzpicture}[scale=0.65, transform shape]
    \node[indonor] (d0) at (0,0) {$d_0$};
    \node[patient] (p1) at (1.5,1) {$p_1$};
    \node[patient] (p2) at (1.5,-1) {$p_2$};
    \node[outdonor] (d1) at (3,1) {$d_1$};
    \node[outdonor] (d2) at (3,-1) {$d_2$};

    \draw[kidneyarcs] 
    (d0) edge (p1) edge (p2)
    (p1) edge (d1)
    (p2) edge (d2)
    ;
\end{tikzpicture}
\hspace{15mm}
\begin{tikzpicture}[scale=0.65, transform shape]
    \node[indonor] (d0) at (0,0) {$d_0$};
    \node[patient] (p) at (1.5,0) {$p_1$};
    \node[outdonor] (d1) at (3,1) {$d_1$};
    \node[outdonor] (d2) at (3,-1) {$d_2$};

    \draw[kidneyarcs] 
    (d0) edge (p)
    (p) edge (d1)
    (p) edge (d2)
    ;
\end{tikzpicture}
\caption{Choice gadget (left) and duplication gadget (right)}\label{fig:choice-gadget-and-duplication-gadget}
\end{figure}

Several such duplication gadgets may be joined to form a general duplication gadget with any fixed number of output nodes.
\begin{lemma}[Duplication gadget]
For any positive integer $k$, there is a gadget with a single input node, $k$ patient nodes\footnote{We use $k$ (rather than $k-1$) patient nodes in order to simplify the presentation and proof of Lemma~\ref{lem:e3sat-gadget}.}, and $k$ output nodes such there is an allocation using the input node, including all patient nodes, and provides all output nodes.
\label{lem:duplication-gadget}
\end{lemma}
We now present our gadget construction based on Theorem~\ref{thm:Haastad-E3SAT}.
\begin{lemma}
\label{lem:e3sat-gadget}
    Given a E3CNF formula $\Phi$ with $n$ variables and $m \geq n$ clauses, we may in polynomial time compute a kidney exchange graph with a single input node, $90m$ output nodes, and $3n+94m$ patient nodes. If $\Phi$ is satisfiable then there is an allocation that uses the input node, includes at least $92m$ patient nodes, and provides all $90m$ output nodes. If at most $(8/9)m$ clauses may be satisfied in $\Phi$, then any maximal allocation that uses the input node includes at most $86m$ patient nodes and provides at most $80m$ output nodes.
\end{lemma}
\begin{proof}
    The gadget is obtained by combining several simple gadgets in~5 groups as follows.
    \begin{enumerate}
    \item The input node of the gadget is given as the input node of a duplication gadget, as provided by Lemma~\ref{lem:duplication-gadget}, with~$n$ outputs with the $i$-th output node corresponding to the variable~$x_i$.
    \item The output corresponding to~$x_i$ is identified with the input node of a choice gadget, as given in Figure~\ref{fig:choice-gadget-and-duplication-gadget} and the two output nodes will correspond to the literals $x_i$ and $\neg x_i$.
    \item The output nodes corresponding to a literal $\ell$ is identified with the input node of a new duplication gadget with as many output nodes as the number of occurrences of~$\ell$ in~$\Phi$.
    \item To every clause we form a clause gadget consisting of 3~input donor nodes with directed edges to a single patient node with a proxy donor that serves as output node. We identify the input nodes of a clause gadget with unique output nodes of the previous duplication gadgets corresponding to the literals of the clause.
    \item Finally we identify the output nodes corresponding to each clause with an input node of new duplication gadget each with $90$ outputs.
    \end{enumerate}
    Clearly the total number of output nodes is $90m$. The total number of patient nodes are $N=n+2n+3m+m+90m=3n+94m$.

    Consider a maximal allocation using the single input node. The choice gadgets then define a truth assignment to the variables, and a patient node corresponding to a clause is contained in the allocation precisely when the clause is satisfied by the truth assignment. Let $k$ denote the number of clauses satisfied. Then precisely $90k$ output nodes are provided by the allocation. All $n$ patients in group~1 are in the allocation and precisely~$n$ of the~$2n$ patient nodes group~2 are in the allocation. In group~3 and group~4 at least $2k$ and at most $3m+k$ patient nodes are in the allocation. In group~5 precisely~$90k$ patient nodes are in the allocation. Thus, in total, at least $2n+92k \geq 92k$ and at most $2n+3m+91k \leq 5m+91k$ patient nodes are in the allocation.

    Thus, if $\Phi$ is satisfiable, there is a maximal allocation such that all $90m$ output nodes are provided and having at least $92m$ patient nodes contained in the allocation. If on the other hand, at most $(8/9)m$ clauses may be satisfied in $\Phi$, then any maximal allocation will provide at most $90 \cdot 8/9m = 80m$ output nodes and have at most $5m+91 (8/9) m = (773/9) m < 86m$ patient nodes contained in the allocation.
\end{proof}
With this in place we now give the hardness amplification construction.
\begin{theorem}
    It is NP-hard to approximate $\DDKE{\infty}{\infty}$ within any constant.
    \label{thm:DDKE-infty-infty-hardness}
\end{theorem}
\begin{proof}
    Let $\rho>0$ be given and let $k>1$ be an integer such that $(86/90)^k \leq \rho$.  From Theorem~\ref{thm:Haastad-E3SAT} it follows that it is NP-hard to distinguish between satisfiable E3SAT formulas and $8/9$-satisfiable E3SAT formulas. Thus, let $\Phi$ be a E3CNF formula with~$n$ variables and~$m$ clauses such that either $\Phi$ is satisfiable or such that at most $(8/9)m$ clauses of $\Phi$ may be satisfied simultaneously. We may assume without loss of generality that $1 \leq n \leq m$. Let $G$ be the gadget constructed in Lemma~\ref{lem:e3sat-gadget} from $\Phi$ with $N=3n+94m$ patient nodes. We then build the final kidney exchange graph as follows. We start with a gadget simulating an altruistic donor, and join the output node of this to the input of $G$. We then proceed in $k-1$ rounds, joining all output nodes to new copies of the gadget $G$. The copies of $G$ are thus arranged in a $90m$-ary tree of depth $k-1$ and with $(90m)^i$ copies of $G$ at the $i$-level, for $i=0,\dots,k-1$. The resulting kidney excahnge graph has $2+\sum_{i=1}^k N^i$ patient nodes, which is of polynomial size, and can be constructed in polynomial time.
    
    Suppose first that $\Phi$ is satisfiable. Matching both patient nodes of the altruistic donor gadget provides the output node. In $G$ there is an allocation using the the input node that provides all output nodes and includes at least $92m$ patient nodes. We may thus use this allocation in all copies of $G$ and thereby obtain an allocation including at least $2+\sum_{i=0}^{k-1} (90m)^i 92m \geq (90m)^k$ patient nodes.

    Suppose next that at most $8/9m$ clauses of $\Phi$ may be satisfied simultaneously. This means that any allocation using the input node of $D$ can provide at most $80m$ output nodes and include at most $86m$ patient nodes. Consider now any maximal matching in the complete graph. In the $(90m)^i$ copies of $G$ at level~$i$, the allocation can use at most $(80m)^i$ of the input nodes. In the copies of $G$ where the input node is used, at most $86m$ patient nodes are included in the allocation, and in the copies of $G$ where no input node is not used, no patient nodes are included in the allocation. This means that in total at most $2+\sum_{i=0}^{k-1} (80)^i 86m \leq (86m)^k = (86/90)^k (90m)^k \leq \rho (90m)^k$ patient nodes are included in the allocation.
    
    It is thus NP-hard to distinguish between the case where $G$ has an allocation including at least $(90m)^k$ patient nodes and the case where any allocation includes at most $\rho (90m)^k$ patient nodes, and the result follows.
\end{proof}

\section{Algorithmic results}
We devote this section to studying polynomial-time approximation algorithms for the $(2,1)$-DDKE problem. The first simple observation is that the optimal $(2,0)$-DDKE solution yields a $1/2$-approximation to $(2,1)$-DDKE.

We remark that the optimal $(2,0)$-DDKE solution can be easily computed in polynomial time by a maximum matching computation on the undirected graph that has the patients as nodes and an undirected edge between two patients (corresponding to a $2$-cycle in our terminology) if some of the proxy donors of the former is compatible to the latter and vice versa.\footnote{Similarly, by applying the algorithm of Abraham et al.\cite{abraham2007clearing} on this graph, we obtain an optimal solution for $(\infty,0)$-DDKE. Then, we can easily see that this yields a $1/2$-approximation to $(\infty,1)$-DDKE.} 

\begin{theorem}
    The optimal $(2,0)$-DDKE solution is a $1/2$-approximation to the optimal $(2,1)$-DDKE solution.
\end{theorem}

\begin{proof}
    Consider an optimal $(2,1)$-DDKE solution. It consists of disjoint structures with two, three, or four patients. Each of these structures contains a $2$-cycle of patients. Hence, considering only these $2$-cycles, we get a valid $(2,1)$-solution which has at least half of the patients of the optimal $(2,1)$-DDKE solution. Then, the optimal $(2,0)$-DDKE solution can only be better, yielding a $1/2$-approximation, too.
\end{proof}

Designing better approximation algorithms is non-trivial. We will do so by exploiting a breakthrough result by Cygan~\cite{cygan2013improved} on $k$-set packing. In $k$-set packing, we are given a universe of elements and a collection of sets of size $k$ (or $k$-sets) over these elements. The objective is to compute a subcollection of as many disjoint $k$-sets as possible. Cygan~\cite{cygan2013improved} presents polynomial-time algorithm that computes a $\frac{3-\varepsilon}{k+1}$-approximate solution for $k$-set packing, for every constant $\varepsilon>0$. Hence, on input an instance of $k$-set packing in which the maximum number of disjoint $k$-sets is $o_k$, Cygan's algorithm returns a subcollection with at least $\frac{3-\varepsilon}{k+1}\cdot o_k$ disjoint $k$-sets.

We will use Cygan's algorithms\footnote{Interestingly, using the previously best known approximation algorithms for $k$-set packing (with approximation ratio $\frac{2-\varepsilon}{k}$) due to Hurkens and Schrijver~\cite{HS} instead of Cygan's algorithm cannot yield any better than $1/2$-approximation.} as subroutines (with constant $\varepsilon>0$) to compute a solution of $(2,1)$-DDKE as follows:
\begin{itemize}
    \item We first compute an optimal $(2,0)$-DDKE solution, i.e., a maximum number of disjoint $2$-cycles. 
    \item We consider the instance of $3$-set packing that has the nodes as elements, and the collection of $3$-sets is comprised by all valid $(2,1)$-DDKE structures with exactly three patients. Notice that any two such structures are disjoint if and only if they do not share any patients. We apply Cygan's $3$-set packing algorithm on this instance and, equivalently, compute a $(2,1)$-DDKE solution that consists of disjoint structures of three patients.
    \item Similarly, we consider the instance of $4$-set packing that has the patients as elements, and the collection of $4$-sets is comprised by all valid $(2,1)$-DDKE structures with exactly four patients. We apply Cygan's $4$-set packing algorithm on this instance and, equivalently, compute a $(2,1)$-DDKE solution that consists of disjoint structures of four patients.
    \item Among the three solutions, we pick the one with the maximum number of patients.
\end{itemize}

\begin{theorem}
The algorithm above computes a solution of $(2,1)$-DDKE that has an approximation ratio higher than $1/2$ when $\varepsilon<5/12$ and approaching $36/67\approx 0.537$ as $\varepsilon$ approaches $0$. 
\end{theorem}

\begin{proof}
Consider an instance of $(2,1)$-DDKE and assume that the optimal solution consists of the following disjoint $(2,1)$-DDKE structures: $o_2$ $2$-cycles, $o_3$ structures with three patients, and $o_4$ structures of four patients. Also, let $\alpha_2$, $\alpha_3$, and $\alpha_4$ be the number of corresponding structures contained in each of the three solutions computed by the algorithm.

Notice that the optimal number of disjoint $2$-cycles is at least $o_2+o_3+o_4$, since each of the disjoint structures in the optimal $(2,1)$-DDKE solution contains a $2$-cycle. Hence,
\begin{align}\label{eq:alpha_2-bound}
    \alpha_2 &\geq o_2+o_3+o_4.
\end{align}
Similarly, the optimal number of disjoint $3$-sets in the $3$-set packing instance solved in the second step of the algorithm is at least $o_3+o_4$ since each of the disjoint structures in the optimal $(2,1)$-DDKE solution with at least three patients contains a triplet of patients that form a $3$-set. Since Cygan's $3$-set packing algorithm has approximation ratio $\frac{3-\varepsilon}{4}$, we have
\begin{align}\label{eq:alpha_3-bound}
    \alpha_3 &\geq \frac{3-\varepsilon}{4}\cdot \left(o_3+o_4\right).
\end{align}
Finally, the optimal number of disjoint $4$-sets in the $4$-set packing instance solved in the second step of the algorithm is at least $o_4$, since each of the disjoint structures with four patients in the optimal $(2,1)$-DDKE solution corresponds to a $4$-set. Since Cygan's $4$-set packing algorithm has approximation ratio $\frac{3-\varepsilon}{5}$, we have
\begin{align}\label{eq:alpha_4-bound}
    \alpha_4 &\geq \frac{3-\varepsilon}{5}\cdot o_4.
\end{align}
Hence, the number of patients served by the algorithm is
\begin{align}\nonumber
    &\#\text{patients} =\max\left\{2\cdot\alpha_2, 3\cdot \alpha_3, 4\cdot \alpha_4\right\}\\\label{eq:all-bounds-together}
    &\geq \max\left\{2\cdot (o_2+o_3+o_4), 3\cdot \frac{3-\varepsilon}{4}\cdot (o_3+o_4), 4\cdot \frac{3-\varepsilon}{5}\cdot o_4\right\}.
\end{align}
We now use the property 
\begin{align}\label{eq:xyz}
\max\{x,y,z\} &\geq \frac{\zeta\cdot x+\eta\cdot y+\theta\cdot z}{\zeta+\eta+\theta},
\end{align}
for every positive reals $x$, $y$, $z$, $\zeta$, $\eta$, and $\theta$. Using 
$x = 2\cdot (o_2+o_3+o_4)$, $y = 3\cdot \frac{3-\varepsilon}{4}\cdot (o_3+o_4)$, $z = 4\cdot \frac{3-\varepsilon}{5}\cdot o_4$, $\zeta = 1$, $\eta = \frac{4}{3\,(3-\varepsilon)}$, and $\theta = \frac{5}{4\,(3-\varepsilon)}$, we get $x\cdot \zeta+y\cdot \eta+z\cdot \theta = 2\cdot o_2+3\cdot o_3+4\cdot o_4=\text{OPT}$ and $\zeta+\eta+\theta=\frac{67-12\varepsilon}{36-12\varepsilon}$ (where $\text{OPT}$ denotes the number of patients served in the optimal solution of the $(2,1)$-DDKE instance). Applying inequality (\ref{eq:xyz}) to the RHS of Equation (\ref{eq:all-bounds-together}), we obtain
\begin{align*}
    \#\text{patients} &\geq \frac{36-12\,\varepsilon}{67-12\,\varepsilon}\cdot \text{OPT}.
\end{align*}
For $\varepsilon<5/12$, we can verify that the approximation ratio of the solution obtained (i.e., the constant in the RHS of this last bound relating $\#\text{patients}$ and $\text{OPT}$) is higher than $1/2$ and approaches $36/67\approx 0.537$ as $\varepsilon$ approaches $0$. 
\end{proof}

\section{Robustness to strategic manipulation by hiding donors}

We now explore whether kidney exchange algorithms can be implemented as truthful mechanisms. We assume that patients act as agents who report their proxy donors to the system. The system then collects the medical data of all participating patients and donors and computes the compatibility graph on which the kidney exchange algorithm runs. A patient could manipulate the system by hiding some of their proxy donors if this increases their chance of getting a kidney. We explore whether there are algorithms (or, better, {\em mechanisms}) that are robust to such strategic manipulation by the patients. We begin with a weak definition of truthfulness.

\begin{definition}
    A kidney exchange mechanism is {\em weakly truthful} if, for every input, no unsatisfied patient can become satisfied by hiding one of their proxy donors.
\end{definition}

As we will see, social welfare maximization is compatible with weak truthfulness, provided that we handle tie-breaking in a consistent manner. Consider a fixed-priority order for patients based on their IDs (the patient with ID equal to 1 has the highest priority). A solution defines an $n$-entry $0/1$ satisfaction vector, with an entry for each patient. A $1$ or $0$ in the entry corresponding to a patient indicates whether the patient is satisfied or not by the solution. Consider the mechanism that computes a solution of maximum social welfare so that the satisfaction vector defined is lexicographically maximum. We will refer to this mechanism as MSW.

\begin{theorem}\label{thm:weakly-truthful}
    Mechanism MSW is weakly truthful.
\end{theorem}

\begin{proof}
Let $M$ be the solution returned by MSW on an input graph $G$, and consider any patient $i$ who is unsatisfied in $M$. Suppose $i$ hides one of their proxy donors, resulting in a subgraph $G'$ of $G$. Since removing a donor and its incident edges cannot create new valid exchanges, every feasible solution in $G'$ is also a feasible solution in $G$. Moreover, because $i$ is unsatisfied in $M$, none of $i$'s proxy donors can donate in $M$, which implies that $M$ itself remains a feasible solution in $G'$. Because MSW selects a solution whose satisfaction vector is uniquely maximum under the strict total order induced by social welfare and lexicographic tie-breaking, and $M$ remains available in the smaller feasible set of $G'$, the solution $M'$ returned by MSW on $G'$ must induce the exact same satisfaction vector as $M$. Hence, $i$ remains unsatisfied in $M'$.
\end{proof}

We are now ready to define the most interesting definition of truthfulness in our setting.
\begin{definition}
    A kidney exchange mechanism is {\em truthful} if any satisfied patient with two proxy donors becomes unsatisfied when hiding one of their proxy donors. 
\end{definition}

We begin with a negative result for deterministic mechanisms.
\begin{theorem}\label{thm:no-truthful}
For $\rho>1/2$, no $\rho$-approximate deterministic mechanism for $(2,1)$-DDKE or $(\infty,1)$-DDKE is truthful.
\end{theorem}

\tikzset{
  pics/combined-example-non-truthfulness/.style={
    code={
        \node[patient] (p1) at ( 0, 1 ) {$p_1$};
        \node[patient] (p2) at ( -2, -1 ) {$p_2$};
        \node[patient] (p3) at ( 2, -1 ) {$p_3$};
        \node[patient] (p4) at (-4, 1) {$p_4$};

        \node[donor] (d1-1) at (-2, 1) {$d_{1,1}$};
        \node[donor] (d2-1) at (0, -1) {$d_{2,1}$};
        \node[donor] (d2-2) at (-4, -1) {$d_{2,2}$};

        \draw[kidneyarcs]
        (d2-2) edge (p4)
        (p2) edge (d2-2)
        ;

        \draw[boldkidneyarcs]
        (p1) edge (d1-1)
        (d1-1) edge (p2)
        (d2-1) edge (p1)
        (p2) edge (d2-1)
        ;
        }
    }
}

\tikzset{
  hidden/.style={draw=gray, text=gray},
  hiddenedge/.style={draw=gray}
}

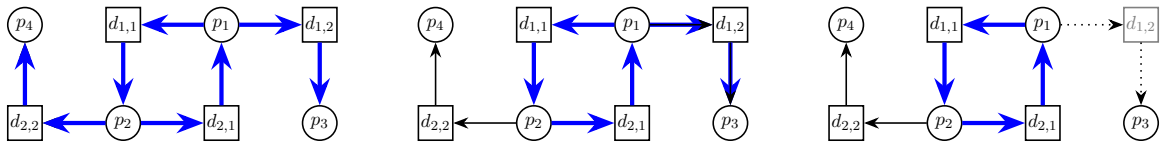
\begin{figure}[h]
\centering
\begin{tikzpicture}[scale=0.65, transform shape]
  \pic {combined-example-non-truthfulness};
    \node[donor] (d1-2) at (2, 1) {$d_{1,2}$};
    \draw[boldkidneyarcs]
    (p1) edge (d1-2)
    (d1-2) edge (p3)
    (p2) edge (d2-2)
    (d2-2) edge (p4)
    ;
\end{tikzpicture}
\hspace{8mm}
\begin{tikzpicture}[scale=0.65, transform shape]
  \pic {combined-example-non-truthfulness};
    \node[donor] (d1-2) at (2, 1) {$d_{1,2}$};
    \draw[boldkidneyarcs]
    (p1) edge (d1-2)
    (d1-2) edge (p3);
    \draw[kidneyarcs]
        (p1) edge (d1-2)  
        (d1-2) edge (p3)
    ;
\end{tikzpicture}
\hspace{8mm}
\begin{tikzpicture}[scale=0.65, transform shape]
  \pic {combined-example-non-truthfulness};
    \node[donor, draw=gray, text=gray] (d1-2) at (2,1) {$d_{1,2}$};    
    \draw[kidneyarcs]
        (p1) edge[dotted] (d1-2)  
        (d1-2) edge[dotted] (p3)
    ;
\end{tikzpicture}
\caption{The argument used in the proof of Theorem~\ref{thm:no-truthful}. An instance of $(2,1)$-DKKE with an optimal solution (left), a $1/2$-approximate solution (middle), and an $1/2$-approximate solution after patient $p_1$ hides donor $d_{1,2}$ (right).}
\label{fig:no-truthful}
\end{figure}

\begin{proof}
Consider the input depicted in Figure~\ref{fig:no-truthful} and notice that the optimal $(2,1)$-DDKE or $(\infty,1)$-DDKE solution satisfies all four patients. Hence, any $1/2$-approximate solution (e.g., the middle solution in Figure~\ref{fig:no-truthful}) must include the $2$-cycle with patients $p_1$ and $p_2$; if this is not the case, no patient could be satisfied. Hence, patient $p_1$ is satisfied. Now, assume that patient $p_1$ hides their proxy donor $d_{1,2}$. The optimal $(2,1)$-DDKE solution is essentially symmetric to the one depicted in the middle of Figure~\ref{fig:no-truthful} that satisfies three patients. Again, any $1/2$-approximate solution (e.g., the right solution in Figure~\ref{fig:no-truthful}) must include the $2$-cycle with patients $p_1$ and $p_2$, and again satisfy patient $p_1$. Hence, patient $p_1$ can hide their proxy donor $d_{1,1}$ and still be satisfied, violating truthfulness.
\end{proof}

For completeness, we present a tight positive result. 
However, it is clearly unsatisfactory since it ignores completely the benefits of the second proxy donor. Denote by MSWC the mechanism defined as follows. On input a kidney exchange graph, for every patient who has two proxy donors, we consolidate them into one. In this way, we obtain a kidney exchange instance in which each patient has a single proxy donor. We apply the mechanism MSW to this new graph. We refer to this mechanism as MSWC. 

The solution of MSWC is a set of disjoint cycles that serves the maximum number of patients. Clearly, by hiding some of their proxy donors, a patient just hides some of the outgoing edges from their consolidated donor and, thus, can only decrease their chance of being part of a cycle in the solution returned by MSWC. The consolidation step ensures that only one proxy donor will be used for each served patient. For instances of $(2,1)$- and $(\infty,1)$-DDKE, the other proxy donor could reach at most one extra patient. Thus, the number of patients that are ``lost'' due to this consolidation is not higher than the number of patients actually served. The next theorem is immediate.

\begin{theorem}
    Mechanism MSWC is a $1/2$-approximate, polynomial-time, deterministic, truthful mechanism for problems $(2,1)$-DDKE and $(\infty,1)$-DDKE.
\end{theorem}

More interestingly, by resorting to randomization, we can (almost) achieve social welfare maximization. We refine the definition of truthfulness as follows.

\begin{definition}
A randomized mechanism is truthful if, for every input, removing a proxy donor of a patient decreases the probability that the patient is served.
\end{definition}

Now, define the randomized mechanism RSW as follows. Let $n$ be the number of patients and $m$ the total number of donors. The mechanism uses a constant $\varepsilon \in \left(0, \frac{1}{2n}\right)$ and is defined as follows. It tosses a coin with $\Pr[\text{heads}] = (2n - m) \cdot \varepsilon$ (since $0 \leq m \leq 2n$, this probability is well-defined):
\begin{itemize}
    \item On heads, it returns an empty solution.
    \item On tails, it runs mechanism MSW on the input graph.
\end{itemize}
Clearly, as $\varepsilon$ approaches $0$, the expected size of the solution returned by mechanism RSW approaches the optimal one. Furthermore, we prove the following statement.

\begin{theorem}
    For every $\varepsilon \in (0,1/n)$, mechanism RSW is truthful.
\end{theorem}

\begin{proof}
Let $M$ be the solution returned by MSW on an input graph $G$ with $n$ patients and $m$ donors, and consider any patient $i$ who hides some of their proxy donors, resulting in a subgraph $G'$ with at most $m - 1$ donors. Let $M'$ be the solution returned by MSW on $G'$. If $i$ is unsatisfied in $M$, then by Theorem~\ref{thm:weakly-truthful}, $i$ remains unsatisfied in $M'$, so the probability that $i$ is served remains $0$. If $i$ is satisfied in $M$, then $i$ is served by RSW on $G$ with probability $1 - (2n - m)\varepsilon > 0$. On $G'$, either $i$ becomes unsatisfied in $M'$ (reducing their probability of being served to $0$) or $i$ remains satisfied in $M'$ (reducing their probability of being served to at most $1 - (2n - m + 1)\,\varepsilon$). In both cases, the probability that $i$ is served decreases.
\end{proof}

\section{Experimental results}
Our experimental evaluation extends the random graph model of Toulis and Parkes \cite{toulis2011random} to accommodate dual proxy donors. For a given set of parameters---total number of patients $n$, number of patients with two proxy donors $k$, cross-match compatibility probability $p_c$, and blood type distribution $P$---we construct compatibility graphs as follows.\footnote{In our experiments, to be consistent with the setting of~\cite{toulis2011random}, we use the same parameters they used, i.e., a cross-match compatibility probability $p_c=0.2$ and the blood type distribution (O: $45\%$, A: $40\%$, B: $11\%$, AB: $4\%$).} The final graph always contains exactly $n$ patient nodes, comprising $k$ patients with two proxy donors and $n-k$ patients with a single proxy donor.

A patient with a single proxy donor is generated by first drawing random blood types for both the patient and the donor according to the distribution $P$. If their blood types are blood-type incompatible, both the patient and the donor (together with the directed edge from the patient to the donor) are added in the graph. If they are blood-type compatible, a simulated cross-match determines actual compatibility: the pair is deemed compatible with probability $p_c$ and incompatible with probability $1-p_c$. The patient and the donor are added to the graph (together with the patient-to-donor directed edge) only if they are ultimately incompatible; otherwise, the entire generation process for this patient-donor repeats from scratch.

For a patient with two proxy donors, blood types for the patient and both donors are drawn from $P$. If the patient is blood-type incompatible with both donors, the patient and the two donors (as well as directed edges from the patient to the two donors) are added in the graph. If the patient is blood-type compatible with one or both donors, we simulate cross-matching for each blood-compatible donor: compatibility is declared with probability $p_c$, and incompatibility with probability $1-p_c$. The patient and the two donors are added to the graph (together with the two patient-to-donor directed edges) only if both donors are ultimately determined to be incompatible with the patient. If either donor is compatible, the entire generation process for this triplet repeats from scratch.

Once all patient and donor nodes are established, we determine the compatibilities between patients and donors. For every patient and donor that is not the patient's proxy, if they are blood-type compatible, a directed edge from the donor to the patient is added with the cross-match probability $p_c$.

We evaluate and compare two specific problem variants: $(2,0)$-DDKE and $(2,1)$-DDKE. We use a matching computation to compute the optimal $(2,0)$-DDKE solution. The optimal $(2,1)$-DDKE solution is computed via solving an integer linear program (ILP) in Python. To optimize graph traversal speed during cycle enumeration, the compatibility graph is constructed using native dictionary-based adjacency lists rather than external graph libraries. The ILP is formulated using the \texttt{PuLP} library and solved via the \texttt{CBC} solver with a strict time limit (e.g., 15 seconds per instance). Statistical processing of the experimental trials is handled via \texttt{NumPy} and \texttt{SciPy}.

\begin{figure*}[t]
    \centering
\resizebox{\textwidth}{!}{
\begin{tikzpicture}
\begin{axis}[
    width=12cm,
    height=8cm,
    title={Satisfied patients vs.~double donor ratio ($n=100$)},
    xlabel={Patients with two proxy donors (\%)},
    ylabel={Satisfied patients (\%)},
    xmin=-1, xmax=21,
    ymin=47.5, ymax=60.8,
    xtick={0, 2, 4, 6, 8, 10, 12, 14, 16, 18, 20},
    ytick={48, 50, 52, 54, 56, 58, 60},
    xticklabel style={
        /pgf/number format/fixed,
        /pgf/number format/fixed zerofill,
        /pgf/number format/precision=0
    },
    grid=major,
    grid style={dashed, gray!40},
    legend pos=north west,
    legend cell align={left},
    tick label style={font=\small},
    title style={font=\large},
]

\addplot [
    draw=none,
    fill=blue,
    fill opacity=0.2,
    forget plot
] coordinates {
    (0, 51.23) (2, 54.58) (4, 56.31) (6, 56.84) (8, 57.95)
    (10, 58.01) (12, 58.62) (14, 59.48) (16, 59.90) (18, 60.21) (20, 59.49)
    (20, 56.71) (18, 57.59) (16, 56.70) (14, 56.72) (12, 55.98)
    (10, 55.19) (8, 55.05) (6, 53.96) (4, 53.09) (2, 51.82) (0, 48.37)
} \closedcycle;

\addplot [
    draw=none,
    fill=red,
    fill opacity=0.2,
    forget plot
] coordinates {
    (0, 51.23) (2, 52.78) (4, 53.11) (6, 52.34) (8, 53.00)
    (10, 52.00) (12, 51.99) (14, 52.36) (16, 52.58) (18, 52.63) (20, 52.40)
    (20, 49.60) (18, 50.17) (16, 49.42) (14, 49.64) (12, 49.21)
    (10, 49.00) (8, 50.00) (6, 49.46) (4, 49.89) (2, 50.02) (0, 48.37)
} \closedcycle;

\addplot [
    color=blue,
    thick,
    mark=*,
] coordinates {
    (0, 49.8) (2, 53.2) (4, 54.7) (6, 55.4) (8, 56.5)
    (10, 56.6) (12, 57.3) (14, 58.1) (16, 58.3) (18, 58.9) (20, 58.1)
};
\addlegendentry{(2,1)-DDKE solution}

\addplot [
    color=red,
    thick,
    mark=square*,
] coordinates {
    (0, 49.8) (2, 51.4) (4, 51.5) (6, 50.9) (8, 51.5)
    (10, 50.5) (12, 50.6) (14, 51.0) (16, 51.0) (18, 51.4) (20, 51.0)
};
\addlegendentry{(2,0)-DDKE solution}

\end{axis}
\end{tikzpicture}

\begin{tikzpicture}
\begin{axis}[
    width=12cm,
    height=8cm,
    title={Satisfied patients vs.~double donor ratio ($n=200$)},
    xlabel={Patients with two proxy donors (\%)},
    ylabel={Satisfied patients (\%)},
    xmin=-1, xmax=21,
    ymin=51, ymax=61,
    xtick={0, 2, 4, 6, 8, 10, 12, 14, 16, 18, 20},
    ytick={52, 54, 56, 58, 60},
    xticklabel style={
        /pgf/number format/fixed,
        /pgf/number format/fixed zerofill,
        /pgf/number format/precision=0
    },
    grid=major,
    grid style={dashed, gray!40},
    legend pos=north west,
    legend cell align={left},
    tick label style={font=\small},
    title style={font=\large},
]

\addplot [
    draw=none,
    fill=blue,
    fill opacity=0.2,
    forget plot
] coordinates {
    (0, 53.60) (2, 55.03) (4, 56.84) (6, 58.35) (8, 59.45)
    (10, 59.57) (12, 60.68) (14, 60.46) (16, 59.56) (18, 59.91) (20, 60.46)
    (20, 58.74) (18, 57.89) (16, 57.64) (14, 58.34) (12, 58.92)
    (10, 57.83) (8, 57.55) (6, 56.25) (4, 54.76) (2, 53.17) (0, 51.60)
} \closedcycle;

\addplot [
    draw=none,
    fill=red,
    fill opacity=0.2,
    forget plot
] coordinates {
    (0, 53.60) (2, 53.24) (4, 53.52) (6, 53.77) (8, 54.05)
    (10, 53.99) (12, 54.75) (14, 54.38) (16, 53.57) (18, 54.42) (20, 54.37)
    (20, 52.43) (18, 52.18) (16, 51.43) (14, 52.22) (12, 52.85)
    (10, 52.21) (8, 52.15) (6, 51.63) (4, 51.48) (2, 51.36) (0, 51.60)
} \closedcycle;

\addplot [
    color=blue,
    thick,
    mark=*,
] coordinates {
    (0, 52.6) (2, 54.1) (4, 55.8) (6, 57.3) (8, 58.5)
    (10, 58.7) (12, 59.8) (14, 59.4) (16, 58.6) (18, 58.9) (20, 59.6)
};
\addlegendentry{(2,1)-DDKE solution}

\addplot [
    color=red,
    thick,
    mark=square*,
] coordinates {
    (0, 52.6) (2, 52.3) (4, 52.5) (6, 52.7) (8, 53.1)
    (10, 53.1) (12, 53.8) (14, 53.3) (16, 52.5) (18, 53.3) (20, 53.4)
};
\addlegendentry{(2,0)-DDKE solution}

\end{axis}
\end{tikzpicture}

\begin{tikzpicture}
\begin{axis}[
    width=12cm,
    height=8cm,
    title={Satisfied patients vs.~double donor ratio ($n=500$)},
    xlabel={Patients with two proxy donors (\%)},
    ylabel={Satisfied patients (\%)},
    xmin=-1, xmax=21,
    ymin=52.8, ymax=60.8,
    xtick={0, 2, 4, 6, 8, 10, 12, 14, 16, 18, 20},
    ytick={53, 54, 55, 56, 57, 58, 59, 60},
    xticklabel style={
        /pgf/number format/fixed,
        /pgf/number format/fixed zerofill,
        /pgf/number format/precision=0
    },
    grid=major,
    grid style={dashed, gray!40},
    legend pos=north west,
    legend cell align={left},
    tick label style={font=\small},
    title style={font=\large},
]

\addplot [
    draw=none,
    fill=blue,
    fill opacity=0.2,
    forget plot
] coordinates {
    (0, 54.44) (2, 56.03) (4, 57.85) (6, 58.20) (8, 58.94)
    (10, 59.27) (12, 59.50) (14, 59.96) (16, 60.34) (18, 60.01) (20, 60.02)
    (20, 58.78) (18, 58.79) (16, 59.06) (14, 58.84) (12, 58.30)
    (10, 57.93) (8, 57.66) (6, 57.00) (4, 56.55) (2, 54.77) (0, 53.16)
} \closedcycle;

\addplot [
    draw=none,
    fill=red,
    fill opacity=0.2,
    forget plot
] coordinates {
    (0, 54.44) (2, 54.53) (4, 55.05) (6, 54.81) (8, 55.07)
    (10, 55.07) (12, 55.00) (14, 55.21) (16, 55.72) (18, 55.17) (20, 55.14)
    (20, 53.86) (18, 54.03) (16, 54.48) (14, 53.99) (12, 53.80)
    (10, 53.73) (8, 53.73) (6, 53.59) (4, 53.75) (2, 53.27) (0, 53.16)
} \closedcycle;

\addplot [
    color=blue,
    thick,
    mark=*,
] coordinates {
    (0, 53.8) (2, 55.4) (4, 57.2) (6, 57.6) (8, 58.3)
    (10, 58.6) (12, 58.9) (14, 59.4) (16, 59.7) (18, 59.4) (20, 59.4)
};
\addlegendentry{(2,1)-DDKE solution}

\addplot [
    color=red,
    thick,
    mark=square*,
] coordinates {
    (0, 53.8) (2, 53.9) (4, 54.4) (6, 54.2) (8, 54.4)
    (10, 54.4) (12, 54.4) (14, 54.6) (16, 55.1) (18, 54.6) (20, 54.5)
};
\addlegendentry{(2,0)-DDKE solution}

\end{axis}
\end{tikzpicture}}
\caption{Fraction of satisfied patients (for instances with $n=100, 200, 500$ patients) by the optimal $(2,0)$-DDKE and $(2,1)$-DDKE solutions) for varying fractions of patients with two proxy donors (i.e., parameter $k$ corresponding from $0\%$ up to $20\%$ of the number of patients). Shaded error bands represent $95\%$ confidence intervals over $100$ independent trials.}
    \label{fig:plots}
\end{figure*}
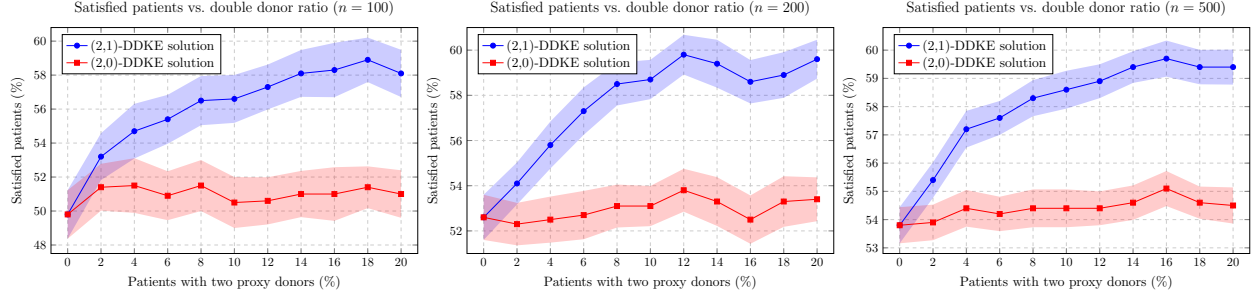

To compute the optimal $(2,1)$-DDKE solution, our ILP has binary variables $x_c$ for each $2$-cycle $c$ and $y_e$ for each single-hop chain $e$. To reduce the search space, we employ a pruning strategy where chain edge variables $y_e$ are only instantiated if the source patient already participates in at least one valid $2$-cycle. The objective is to maximize $2 \sum_c x_c + \sum_e y_e$. The constraints enforce that each patient receives at most one incoming kidney from either a cycle or a chain, each donor gives at most one kidney (either in a cycle or a chain), and a chain can only be initiated by a donor if their associated patient is simultaneously participating in a selected $2$-cycle.

To evaluate the real-world applicability of our model, we generated instances comparable in size to actual matching runs in the UK and the US, specifically testing graphs with $n=100$, $200$, and $500$ patients. In these experiments, the fraction of patients possessing two proxy donors was varied from 0\% to 20\%. To ensure statistical robustness, we conducted $100$ independent trials for each data point, allowing us to compute $95\%$ confidence intervals for the mean satisfaction rates. 

The experimental results, including the $95\%$ confidence intervals illustrated via shaded error bands, are depicted in Figure \ref{fig:plots}. When $k=0$, the optimal solutions for both $(2,0)$-DDKE and $(2,1)$-DDKE are naturally identical. However, as the fraction of patients with two proxy donors increases, a clear divergence emerges between the two approaches:
\begin{itemize}
    \item $(2,0)$-DDKE: The proportion of served patients remains relatively stagnant, fluctuating slightly above $50\%$ across all instance sizes. This indicates that restricting the exchange to pure $2$-cycles fails to capitalize on the presence of dual donors.
    \item $(2,1)$-DDKE: The solution quality improves significantly as the parameter $k$ increases. When $20\%$ of the patients have two proxy donors, the optimal $(2,1)$-DDKE solution successfully serves roughly $58\%$ to $60\%$ of all patients. 
\end{itemize}
This performance gap represents a $9\%$ to $14\%$ relative increase in the total number of patients served. Notably, the growth curve for matched patients is steepest at the smaller fractions of dual-donor participation, demonstrating immediate high returns. Overall, our experiments strongly suggest that convincing more patients to enroll in the system with two donors instead of one can yield substantial systemic benefits, improving solution outcomes for the entire patient population.

\section{Conclusion and open problems}

In this paper, we formulated and studied the dual-donor kidney exchange (DDKE) model, demonstrating that allowing up to two proxy donors per patient to donate simultaneously gives rise to a computationally and strategically rich framework. Our work opens several natural directions for future research. From a theoretical perspective, a notable gap remains for $(2,1)$-DDKE between our $0.537$ approximation ratio and our APX-hardness bound, making the design of improved approximation algorithms an intriguing open problem. Equally compelling is the $(\infty,1)$-DDKE problem: while computing a maximum cycle cover solves $(\infty,0)$-DDKE in polynomial time and immediately yields a baseline $1/2$-approximation for $(\infty,1)$-DDKE, set-packing techniques do not carry over to unbounded cycles, leaving open whether a polynomial-time algorithm can achieve an approximation ratio strictly better than $1/2$. 

Above all, our findings point to a promising practical opportunity. Motivated by our experiments, which demonstrate substantial gains in total transplants even when only a small fraction of patients enroll with two donors, we believe that adapting real-world kidney exchange programs to support dual-donor participation could meaningfully expand the effective organ supply and improve outcomes for patients awaiting a transplant.

\bibliographystyle{plain}
\bibliography{kidney}

\end{document}